\documentclass[conference]{IEEEtran}
\usepackage{amsmath,amssymb,amsfonts, mathrsfs}
\usepackage{amsthm}
\usepackage{bbm}
\usepackage[linesnumbered,ruled]{algorithm2e}
\usepackage{array}
\ifCLASSOPTIONcompsoc
    \usepackage[caption=false, font=normalsize, labelfont=sf, textfont=sf]{subfig}
\else
\usepackage[caption=false, font=footnotesize]{subfig}
\fi
\usepackage{textcomp}

\usepackage{graphicx}
\usepackage{epstopdf}
\usepackage{cite}
\usepackage{booktabs}
\usepackage{multirow}
\usepackage{bigstrut}
\usepackage{color}

\graphicspath{{images/}}
\usepackage{subfiles}

\newtheorem{theorem}{Theorem}
\newtheorem{lemma}{Lemma}
\newtheorem{remark}{Remark}

\begin{document}
\IEEEoverridecommandlockouts

\title{Goal-Oriented Estimation of Multiple Markov Sources in Resource-Constrained Systems}
\author{
\IEEEauthorblockN{Jiping~Luo~and~Nikolaos~Pappas}
    \IEEEauthorblockA{Department of Computer and Information Science, Link\"oping University, Sweden}
\IEEEauthorblockA{Email:~\{jiping.luo,~nikolaos.pappas\}@liu.se}

\thanks{This work has been supported in part by the Swedish Research Council (VR), Excellence Center at Linköping – Lund in Information Technology (ELLIIT), Graduate School in Computer Science (CUGS), the European Union (ETHER, 101096526), and the European Union's Horizon Europe research and innovation programme under the Marie Skłodowska-Curie Grant Agreement No 101131481 (SOVEREIGN).}
}

\maketitle

\begin{abstract}
	This paper investigates goal-oriented communication for remote estimation of multiple Markov sources in resource-constrained networks. An agent decides the updating times of the sources and transmits the packet to a remote destination over an unreliable channel with delay. The destination is tasked with source reconstruction for actuation. We utilize the metric \textit{cost of actuation error} (CAE) to capture the state-dependent actuation costs. We aim for a sampling policy that minimizes the long-term average CAE subject to an average resource constraint. We formulate this problem as an average-cost constrained Markov Decision Process (CMDP) and relax it into an unconstrained problem by utilizing \textit{Lyapunov drift} techniques. Then, we propose a low-complexity \textit{drift-plus-penalty} (DPP) policy for systems with known source/channel statistics and a Lyapunov optimization-based deep reinforcement learning (LO-DRL) policy for unknown environments. Our policies significantly reduce the number of uninformative transmissions by exploiting the timing of the important information.
\end{abstract}

\section{Introduction}

Networked control systems (NCSs) are spatially distributed systems where plants, sensors, controllers, and actuators are interconnected via a shared resource-constrained communication network\cite{walsh2001scheduling, WNCSSurvey}. Such systems are ubiquitous in various applications, such as swarm robotics, connected autonomous vehicles, and smart factories. One of the fundamental problems in these systems is remote estimation of stochastic processes using under-sampled and delayed measurements\cite{NikolaosGoalOriented}. 

Despite various endeavors, most existing communication protocols for remote estimation and control in NCSs are context-agnostic. The primary objective has been to minimize the estimation error (i.e., the \textit{distortion} between the source state and the reconstructed state)\cite{RE-2, Cocco2023, pezzutto2022transmission}, indicating that information is valuable when it is accurate at the point of actuation. However, high accuracy does not necessarily mean better actuation performance. Consider, for example, a remotely controlled drone communicating with the remote center to ensure safe and successful operation. Due to resource constraints, the drone can only send its observation intermittently. In this context, its status should be updated more frequently in critical situations (e.g., close to an obstacle or deviate from planned trajectories), even though estimation error can be large sometimes. Therefore, it is crucial to factor into the communication process the \textit{semantics} (i.e., state-dependent significance, context-aware requirements, and goal-oriented usefulness) of messages and prioritize the information flow efficiently according to application demands\cite{Marios, Petar}.

Information \textit{freshness}, measured by the Age of Information (AoI), that is, the time elapsed since the latest received update was generated\cite{pappas2023age}, has recently been employed in NCSs\cite{AoI_estimation, WiSwarm, kutsevol2023experimental}. However, AoI does not consider the source evolution and the application context. Several metrics have been introduced to address the shortcomings of AoI\cite{niko2019statechange, AoII_TWC, UoI_estimation, WiOpt23}. The authors in \cite{niko2019statechange, WiOpt23} defined state-dependent AoI variables to account for the significance of different states of the stochastic process. Age of Incorrect Information (AoII)\cite{AoII_TWC}, defined as a composite of distortion and age penalties, was employed to capture the cost of not having a correct estimate for some time. The Urgency of Information (UoI)\cite{UoI_estimation} is a weighted distortion metric incorporating context-aware significance through weights. However, these metrics do not directly capture the ultimate goal of communication in NCSs --- actuation. In \cite{AoA23} the authors defined the Age of Actuation (AoA) which is a more general metric than AoI and becomes relevant
when the information is utilized to perform actions in a timely manner.

This paper extends the results of \cite{NikolaosGoalOriented, EmmanouilMinizationCAE, SalimnejadTCOM}. A semantic-empowered and goal-oriented metric, namely \textit{cost of actuation error} (CAE), was first introduced in \cite{NikolaosGoalOriented} to capture state-dependent actuation costs. The problem of remote tracking of a discrete-time Markov source in resource-constrained systems was further studied in \cite{EmmanouilMinizationCAE, SalimnejadJCN, SalimnejadTCOM}. In this work, we consider a more general case where an agent observes multiple Markov sources and decides when to update source status to minimize the long-term average CAE while satisfying an average cost constraint. In addition, we consider a one-slot communication delay between the transmitter and the receiver, making the problem more realistic and challenging. This problem is formulated as a constrained Markov Decision Process (CMDP) and is relaxed using the Lyapunov optimization theorem. We propose a low-complexity drift-plus-penalty (DPP) policy for known environments and a learning-based policy for unknown environments. Our policies achieve near-optimal performance in CAE minimization and significantly reduce uninformative transmissions.


\section{System Model and Problem Formulation}\label{sec:system_model}
\subsection{System Model}
We consider a slotted-time communication system shown in Fig.~\ref{fig:systemmodel}, where the destination is tasked with the remote estimation of $M$ Markov sources. Denote $\mathcal{M} = \{1, 2, \ldots, M\}$ as the index set of the sources. Each source $m \in \mathcal{M}$ is modeled by a $N_m$-state discrete-time Markov process\footnote{DTMC is widely applied in safety-critical systems, such as autonomous driving and cyber-physical security\cite{althoff2011,ye2004robustness}.} $\{X_t^m\}_{t\geq 0}$. The value of $X_t^m$ is chosen from a finite set $\mathbb{X}^m = \{1, 2, \ldots, N_m\}$. The state transition matrix of source $m$ is denoted by $P^m$, i.e., for any $i, j \in \mathbb{X}^m$, the state transition probability is $P^m_{i,j} = \mathbb{P}(X_{t+1}^m = j|X_t^m = i)$. 
\begin{figure}[t]
	\centering
	\includegraphics[width=0.9\linewidth]{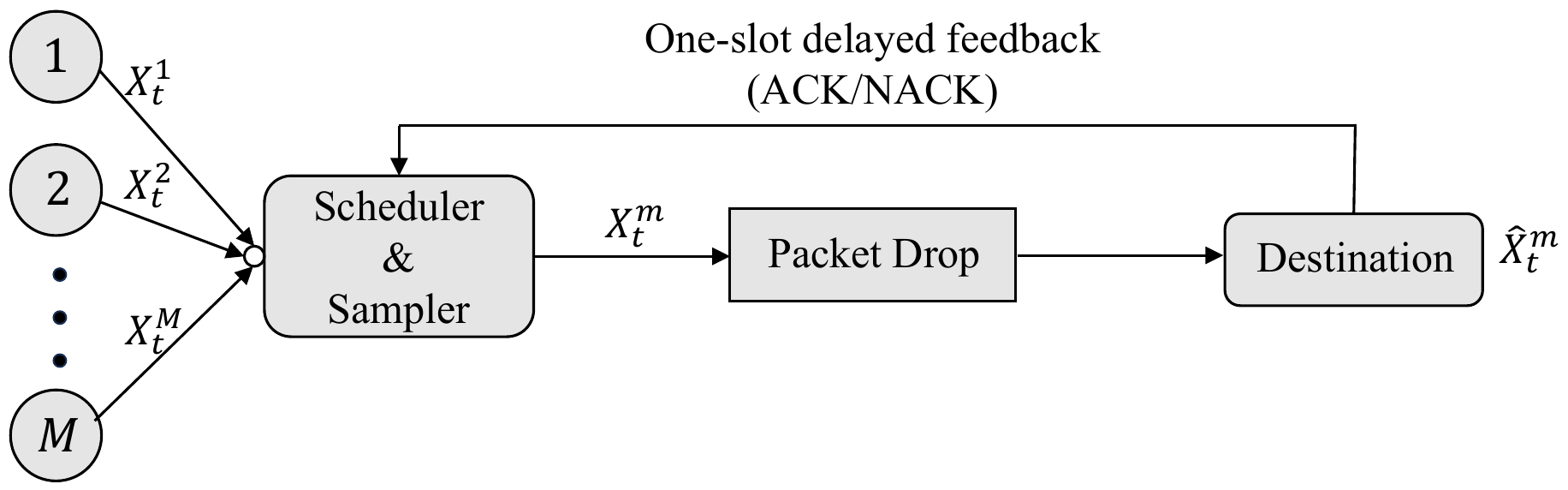}
	\caption{Remote state estimation of multiple Markov sources.}
	\label{fig:systemmodel}
\end{figure}

An agent decides at the beginning of each time slot $t$ which source to sample and then transmits the packet to the destination over an unreliable channel. Only one packet can be sent at a time. The sampling decision at time $t$ is denoted by $\alpha_t$ and is chosen from a finite set $\mathcal{A} = \{0, 1, \ldots, M\}$. If $\alpha_t = m$, it means that source $m$ is sampled, and if $\alpha_t = 0$, it means no source is selected and the transmitter remains silent. Moreover, we consider a cost $c_m$ to represent the resource utilization cost (e.g., bandwidth, and/or power consumption) in sampling and transmission of an update of source $m$. Thus, the cost of performing action $\alpha_t$ can be given by
\begin{align}
	C_t = \sum_{m \in \mathcal{M}} c_m  \mathbbm{1}(\alpha_t = m)\label{eq:sampling and trans cost}
\end{align}
where $\mathbbm{1}(\cdot)$ is the indicator function.

The channel realization $h_t$ is equal to $1$ if the packet is successfully decoded at the receiver and $0$ otherwise. The probability of a successful transmission is defined as $p_s = \mathbb{P}(h_t = 1)$. We consider that the sampling and transmission processes take one slot to be performed. Therefore, the receiver receives an update from the transmitter after a one-slot delay. After each transmission, the receiver sends an acknowledgment (ACK)/NACK packet to the transmitter to indicate whether the transmission was successful or not. It is assumed that ACK packets can be delivered instantaneously and error-free. This is a common assumption in the literature, as ACK/NACK packets are generally much smaller than data packets and are possibly sent over a separate channel. 

The source states are reconstructed at the destination using the following estimate
\begin{align}
	\hat{X}_{t+1}^m = \begin{cases}
		X_{t}^m, & \mathrm{if}~\alpha_{t} = m~\mathrm{and}~h_{t} = 1, \\
		\hat{X}_{t}^m, & \mathrm{otherwise}.
	\end{cases} \label{eq:estimate}
\end{align}

The actuator takes actions according to the estimated state of the sources, i.e., $u_t = \phi(\hat{X}_t^1, \ldots, \hat{X}_t^{M})$. The discrepancy between the source states and the reconstructed states can cause actuation errors. In practice, some source states are more critical than others, and thus some actuation errors may have a larger impact than others. To this end, we utilize the \textit{cost of actuation error} (CAE) to capture the \textit{significance} of error at the point of actuation\cite{NikolaosGoalOriented}. The CAE of each source $m$ can be represented using a pre-defined non-commutative function\footnote{Function $\delta^m$ is non-commutative, i.e., for any $i, j \in \mathbb{X}^m$, $\delta^m (X_t^m = i, \hat{X}_t^m = j)$ is not necessarily equal to $\delta^m (X_t^m = j, \hat{X}_t^m = i)$. This is because different actuation errors may have different costs. By contrast, distortion is generally a commutative measure.} $\delta^m(X_t^m, \hat{X}_t^m)$. The total CAE at time $t$ can be defined as
\begin{align}
		\Delta_t 
		= g(s_t, \alpha_t, s_{t+1})
		=\sum_{m \in \mathcal{M}} \omega_m  \delta^m(X_{t + 1}^m, \hat{X}_{t + 1}^m) \label{eq:CAE}
\end{align}
where $\omega_m \in \mathbb{R}^{+}$ represents the \textit{significance} of source $m$, $s_t = \{X_t^m, \hat{X}_t^m\}_{m\in \mathcal{M}}$ is the system state at time $t$. 
\begin{remark}
    The system has a two-level significance, namely the source significance and the state importance.
\end{remark}

\begin{remark}
    The agent can only evaluate the effectiveness of the decision after a delay of one slot. This implies that the agent needs to make predictions on system state transitions and actuation errors.
\end{remark}

\subsection{Problem Formulation}
For any sampling policy $\pi \triangleq (\alpha_1, \alpha_2, \ldots)$, the time-averaged CAE, denoted by $\bar{\Delta} (\pi)$, and the time-averaged resource utilization costs, denoted by $\bar{C}(\pi)$, are defined as
\begin{align}
	\bar{\Delta}(\pi) \triangleq \limsup\limits_{T\rightarrow \infty} \frac{1}{T} \sum_{t=0}^{T-1} \mathbb{E} \{\Delta_t\}, \label{eq:time-average objective} \\ 
	\bar{C}(\pi) \triangleq \liminf\limits_{T\rightarrow \infty} \frac{1}{T} \sum_{t=0}^{T-1} \mathbb{E} \{C_t\}. \label{eq:time-average constraint}
\end{align}

The goal of our work is to find an optimal sampling policy $\pi$ that minimizes the long-term average CAE of the multi-source system while satisfying an average resource constraint. The stochastic optimization problem can be formulated as
\begin{align}
	\min_{\pi \in \Pi_\textrm{S}}~\bar{\Delta}(\pi), ~\textrm{s.t.}~ \bar{C}(\pi) \leq C_{\textrm{max}} \label{problem:CMDP}
\end{align} 
where $C_\textrm{max} > 0$ is the threshold on the average cost, $\Pi_\textrm{S}$ is the set of all \textit{stationary policies}. For unit cost, i.e., $c_m = 1$, $\bar{C}$ represents the transmission frequency and thus $C_\textrm{max} \in (0, 1]$. Problem \eqref{problem:CMDP} is an \textit{average-cost constrained Markov Decision Process} (CMDP), which is, however, challenging to solve since it imposes a global constraint that involves the entire decision-making process. 

Classical approaches to CMDPs include linear programming and Lagrangian dynamic programming \cite{CMDP_BOOK}. Although optimal results can be achieved, they scale poorly due to the curse of dimensionality. Moreover, these solutions require full knowledge of the system statistics. 

To address these issues, in the following section, we propose a learning-based policy that tackles unknown environments and enables real-time decision-making. 


\section{Proposed Approach}
In this section, we first introduce two methods, namely \textit{Lagrangian relaxation} and \textit{Lyapunov drift}, to transform the CMDP problem \eqref{problem:CMDP} into an unconstrained problem. Then, we propose two policies to solve the relaxed one. 
\subsection{Problem Transformation}
\subsubsection{Lagrangian relaxation method} The constraint in \eqref{problem:CMDP} can be relaxed using \textit{Lagrangian multiplier}, and the resulting \textit{Lagrangian cost function} is defined as
\begin{align}
	\mathcal{L}(\pi, \lambda) \triangleq \underbrace{\bar{\Delta}(\pi)}_\textrm{objective} + \lambda  
	\underbrace{(\bar{C}(\pi) - C_\textrm{max})}_\textrm{penalty}
\end{align}
where $\lambda \geq 0$ is the Lagrangian multiplier that penalizes infeasible solutions. When $\lambda = 0$, it means that data communication is cost-free and the problem is reduced to an unconstrained MDP.

By \cite[Theorem~3.6]{CMDP_BOOK}, the CMDP can be converted into an equivalently unconstrained problem, i.e., 
\begin{align}
	\min_{\pi \in \Pi_\textrm{S}} \sup_{\lambda \geq 0} \mathcal{L}(\pi, \lambda) = \sup_{\lambda \geq 0}\min_{\pi \in \Pi_\textrm{D}}  \mathcal{L}(\pi, \lambda) \label{Problem:LagrangianMDP}
\end{align}
where $\Pi_\textrm{D}$ is the set of all \textit{stationary deterministic policies}, and the right-hand problem is a \textit{Lagrangian MDP}. For any fixed value of $\lambda$, the optimal policy of the Lagrangian MDP, denoted by $\pi^*_\lambda$, is called a $\lambda$-optimal policy. By \cite[Theorem~12.7]{CMDP_BOOK}, the optimal policy of the CMDP is a randomized mixture of two deterministic optimal policies to the Lagrangian MDP, i.e., $\pi^* = \beta \pi^*_{\gamma - \xi} + (1-\beta)\pi^*_{\gamma + \xi}$, where $\gamma = \inf\{\lambda:\bar{C}(\pi^*_\lambda) \leq C_\textrm{max}\}$, $\beta = (C_\textrm{max} - \Bar{C}(\pi^*_{\gamma + \xi})/(\Bar{C}(\pi^*_{\gamma - \xi} - \Bar{C}(\pi^*_{\gamma + \xi})$ is the randomization factor, and $\xi$ is a small perturbation.

However, finding $\gamma$ and the optimal policy is computationally intractable\cite[Section~3.2]{DifficutyOfCMDP}. One practical solution to the CMDP is based on the \textit{value iteration algorithm} (VIA) and the bisection search\cite{EmmanouilMinizationCAE}. Specifically, this approach involves an iterative procedure where VIA is applied at each iteration to find a $\lambda$-optimal policy for a given $\lambda$, and the bisection method is used to update the parameter $\lambda$. Although optimal performance can be achieved, this method is computationally inefficient, especially when dealing with large state/action spaces and multiple sources.

\subsubsection{Lyapunov drift method} According to Lyapunov optimization theorem \cite[Chapter~4]{LyapunovOptimmization}, time-averaged constraints of a stochastic optimization problem can be enforced by transforming them into queue stability problems. Specifically, we define a \textit{virtual queue} $\mathcal{Z}_t$ for the constraint in \eqref{problem:CMDP}, with update equation
\begin{align}
	\mathcal{Z}_{t+1} = \max[\mathcal{Z}_t - C_\textrm{max}, 0] + C_t. \label{eq:virtual queue}
\end{align}
Herein, $C_\textrm{max}$ acts as a virtual service rate and $C_t$ acts as a virtual arrival process. If the virtual queue $\mathcal{Z}_t$ is \textit{mean rate stable}, then the constraint in \eqref{problem:CMDP} is satisfied with probability 1\cite{LyapunovOptimmization}.


To stabilize the virtual queue, we first utilize the \textit{one-slot conditional Lyapunov drift}, which is defined as the expected change in the \textit{Lyapunov function} over one slot given the current system state, i.e.,
\begin{align}
	D(\mathcal{Z}_t) \triangleq \mathbb{E} \{L(\mathcal{Z}_{t+1}) - L(\mathcal{Z}_{t})|\mathcal{Z}_t\} \label{eq:drift function}
\end{align}
where $L(\mathcal{Z}_t) = \frac{1}{2}\mathcal{Z}_t^2$ is a quadratic Lyapunov function The expectation is with respect to the (possibly random) sampling actions. By using the inequality $(\max[Q-b, 0] + A)^2 \leq Q^2 + A^2 + b^2 + 2Q(A-b)$, the upper bound of the Lyapunov drift can be derived as 
\begin{align}
	D(\mathcal{Z}_t) \leq B + \mathbb{E} \bigl\{\mathcal{Z}_t (C_t - C_\textrm{max}) |\mathcal{Z}_t\bigr\} \label{eq:bound on drift}
\end{align}
where $B \geq \mathbb{E}\bigl\{\frac{C_t^2 + C_\textrm{max}^2}{2}|\mathcal{Z}_t\bigr\}$ is a finite constant.

We can utilize the \textit{drift-plus-penalty} (DPP) method to stabilize virtual queues (\textit{drift term}) while minimizing the time-averaged cost (\textit{penalty term}). Specifically, DPP seeks to minimize the upper bound on the following expression 
\begin{align}
	\underbrace{\mathbb{E} \{L(\mathcal{Z}_{t+1}) - L(\mathcal{Z}_{t})|\mathcal{Z}_t\}}_\textrm{drift} + V  \underbrace{\mathbb{E} \{\Delta_t|\mathcal{Z}_t\}}_\textrm{penalty} \label{eq:dpp-1}
\end{align}
where $V$ is a non-negative weight that represents how much emphasis we put on CAE minimization. Notice that the expectation of the penalty term is with respect to all the system randomness, including source state transitions, channel states, and sampling actions. By substituting \eqref{eq:bound on drift} into \eqref{eq:dpp-1}, the upper bound of the drift-plus-penalty expression can be derived as
\begin{align}
	B + \mathbb{E}
	\bigl\{
	\mathcal{Z}_t  (C_t - C_\textrm{max}) 
	+ V \Delta_t | \mathcal{Z}_t
	\bigr\}. \label{eq:bound on dpp}
\end{align}

The DPP policy utilizes the method of \textit{opportunistically minimizing an expectation} \cite[Chapter~3.1]{LyapunovOptimmization} to minimize expression \eqref{eq:dpp-1}. More specifically, at each time $t$, the agent maintains the virtual queue $\mathcal{Z}_t$, observes the system state $s_t$, and takes an action $\alpha_t$ by solving the following problem
\begin{align}
	\min_{\alpha_t \in \mathcal{A}} ~ \mathcal{Z}_t (C_t - C_\textrm{max}) + V \Delta_t. \label{problem:DPP}
\end{align}

\begin{remark}
    The DPP policy is an online policy that has no access to $s_{t+1}$ at time $t$. Therefore, the agent should know a priori the expected costs of taking an action in a certain state. To this end, we replace $\Delta_t$ with the one-slot expected CAE $\bar{\Delta}_t$, as summarized in Lemma \ref{Lemma:expected CAE}. 
\end{remark}

\begin{lemma}\label{Lemma:expected CAE} 
The one-slot expected CAE $\bar{\Delta}_t$ is given by
	\begin{align}
		\bar{\Delta}_t &= \sum_{m \in \mathcal{M}} \omega_m  \bar{\delta}^m_t \text{ where } \\
		\bar{\delta}^m_t &= 
		\begin{cases}
			\sum\limits_{k\neq i} \delta^m_{k,i} P_{i,k}^m p_s 
			+ 
			\sum\limits_{k\neq j} \delta^m_{k,j} P_{i,k}^m (1-p_s), & \textrm{if}~\alpha_t = m \\
			\sum\limits_{k\neq j} \delta^m_{k,j} P_{i,k}^m, & \textrm{if}~\alpha_t \neq m.
		\end{cases}
	\end{align}
\end{lemma}
\begin{proof}
	Assume that at time $t$ the sub-system associated with source $m$ is $s^m_t = (i,j)$. The transition probabilities of sub-system $m$ can be given by
	\begin{align}
		P(s_{t+1}^m| s_t^m, \alpha_t) = \begin{cases}
			P_{i,k}^m p_s, &\textrm{if}~\alpha_t = m, h_t = 1\\
			P_{i,k}^m (1-p_s),&\textrm{if}~\alpha_t = m, h_t = 0\\
			P_{i,k}^m ,&\textrm{if}~\alpha_t \neq  m \\
			0,&\textrm{otherwise}
		\end{cases} \label{eq:transitionProbabilities}
	\end{align}
	where $k \in \mathbb{X}^m$ is the next state of source $m$. Taking expectations over the source/channel randomnesses yields
	\begin{align}
		\bar{\delta}^m_t = \sum_{s^m_{t+1}\in \mathcal{S}^m}P(s_{t+1}^m| s_t^m, \alpha_t)\delta^m(X^m_{t+1}, \hat{X}^m_{t+1}). \label{eq:sub-systemCAE}
	\end{align}
	Substituting \eqref{eq:estimate} and \eqref{eq:transitionProbabilities} into \eqref{eq:sub-systemCAE} and taking weighted sum of all the sub-systems, the lemma is hereby proved.
\end{proof}

The algorithm is given in pseudo-code in Algorithm~\ref{alg:dpp}. The time complexity of the DPP policy is $\mathcal{O}(|\mathcal{A}|)$, thus it has low complexity and it can support large-scale systems. Moreover, the DPP policy satisfies the time average constraint, as proved in Theorem \ref{Theorem:stability}. However, the DPP policy is sub-optimal because: 1) Lyapunov drift primarily focuses on constraint satisfaction thus may result in sub-optimal performance, 2) it minimizes \eqref{eq:dpp-1} in a greedy manner and ignores the long-term system performance. Furthermore, it needs to know a \textit{priori} the channel/source statistics to compute the one-slot expected costs. One can apply an estimation of the system statistics and then use this approach; however it comes with a penalty on the performance that will depend on the accuracy of the estimate.

\begin{algorithm}[htbp]
	\label{alg:dpp}
	\DontPrintSemicolon
	\SetAlgoLined
	\caption{Low-complexity DPP policy.}
	Set $V$, and initialize virtual queue $\mathcal{Z}_0 = 0$\;
	\For{$t = 1, 2, \ldots$}{
		\textsc{Observe} $S_t$ and $\mathcal{Z}_t$ at the beginning of slot $t$.\; 
        \textsc{Calculate} the expected cost $\bar{\Delta}_t$ using Lemma \ref{Lemma:expected CAE}.\;
		\textsc{Select} the best action by solving
        $a^* = \arg\min_{\alpha_t \in \mathcal{A}} \mathcal{Z}_t (C_t - C_\textrm{max}) + V \bar{\Delta}_t$.\;
		\textsc{Apply} $a^*$ to the system and update $\mathcal{Z}_t$.\;
	}
\end{algorithm}

\begin{theorem}\label{Theorem:stability}
	For any non-negative constant $V$, the DPP policy stabilizes the virtual queue $\mathcal{Z}_t$, thereby satisfying the constraint in \eqref{problem:CMDP}.
\end{theorem}
\begin{proof}
	We first consider a special stationary and randomized policy $\alpha_t^*$ that takes actions independent of system state and queue backlog, i.e.,
	\begin{align}
		\alpha_t^* = \begin{cases}
			0, & \textrm{with probability}~ 1- \sum_{m=1}^{M}\frac{C_\textrm{max} - \epsilon}{Mc_m}\\
			m, & \textrm{with probability}~\frac{C_\textrm{max} - \epsilon}{Mc_m}
		\end{cases} \label{eq:source-agnostic}
	\end{align}
	where $0 \leq  \epsilon\leq C_\textrm{max}$. Therefore, the following Slackness condition holds:
	\begin{align}
		\mathbb{E}\{C_t(\alpha_t^*) - C_\textrm{max}\} = -\epsilon. \label{eq:slacknesscondition}
	\end{align}
	Then, the drift-plus-penalty in \eqref{eq:dpp-1} satisfies
	\begin{align}
		&D(\mathcal{Z}_t) + V\mathbb{E}\{\Delta_t|\mathcal{Z}_t\} \notag\\
		\leq& B + V\mathbb{E}\{\Delta_t(\alpha_t)|\mathcal{Z}_t\} + \mathbb{E}\{\mathcal{Z}_t(C_t(\alpha_t) - C_\textrm{max})|\mathcal{Z}_t\} \\
		\overset{(a)}{\leq}& B + V\mathbb{E}\{\Delta_t(\alpha^*_t)|\mathcal{Z}_t\} + \mathbb{E}\{\mathcal{Z}_t(C_t(\alpha^*_t) - C_\textrm{max})|\mathcal{Z}_t\} \\
		\overset{(b)}{\leq}& B + V\Delta_\textrm{max} +  \mathcal{Z}_t\mathbb{E}\{C_t(\alpha^*_t) - C_\textrm{max}\} \label{eq:lastInequality}
	\end{align}
	where \textit{(a)} holds because the DPP policy chooses the best action in set $\mathcal{A}$, including $\alpha_t^*$; \textit{(b)} holds because $\alpha_t^*$ is independent of queue backlog and $\Delta_t$ is upper bounded by $\Delta_\textrm{max}$. Substituting \eqref{eq:slacknesscondition} into \eqref{eq:lastInequality} yields:
	\begin{align}
		D(\mathcal{Z}_t) + V\mathbb{E}\{\Delta_t|\mathcal{Z}_t\} \leq B + V\Delta_\textrm{max} - \epsilon\mathcal{Z}_t. \label{eq:LyapunovForm}
	\end{align}
	The above expression is in the exact form of the Lyapunov optimization theorem\cite[Theorem~4.2]{LyapunovOptimmization}. Therefore, $\mathcal{Z}_t$ is mean rate stable, and the time average constraint is satisfied.
\end{proof}

\begin{figure*}[t]
	\centering
	\begin{minipage}[t]{0.49\textwidth}
		\centering
		\includegraphics[width=0.85\linewidth]{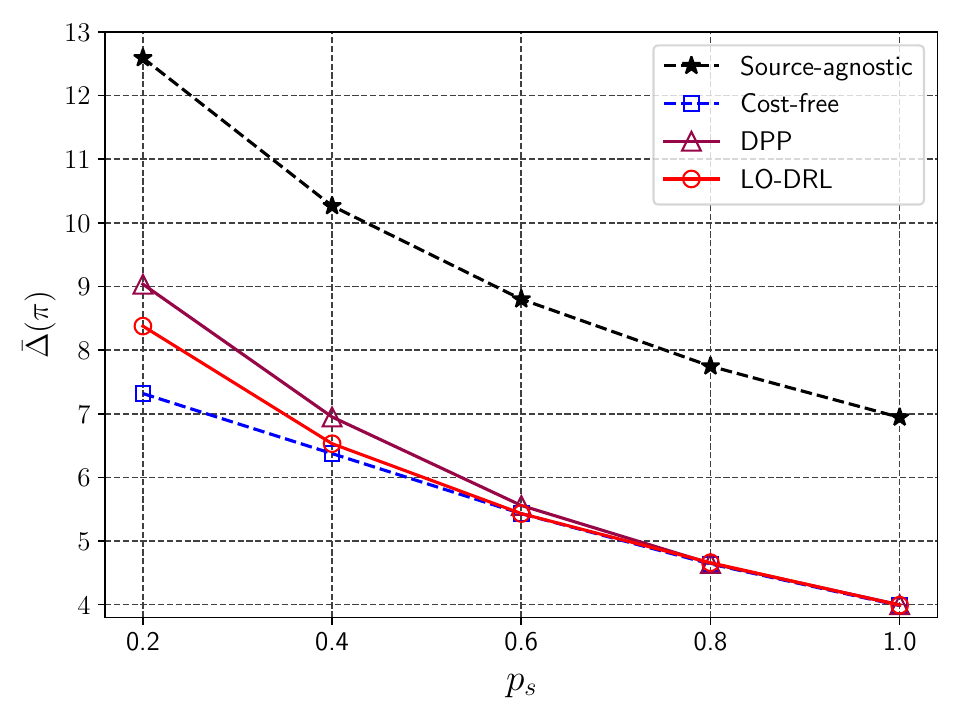}
		\caption{The average CAE vs. $p_s$ for the different policies.}
		\label{fig:costdiffps}
	\end{minipage}
	\hfill
	\begin{minipage}[t]{0.49\textwidth}
		\centering
		\includegraphics[width=0.85\linewidth]{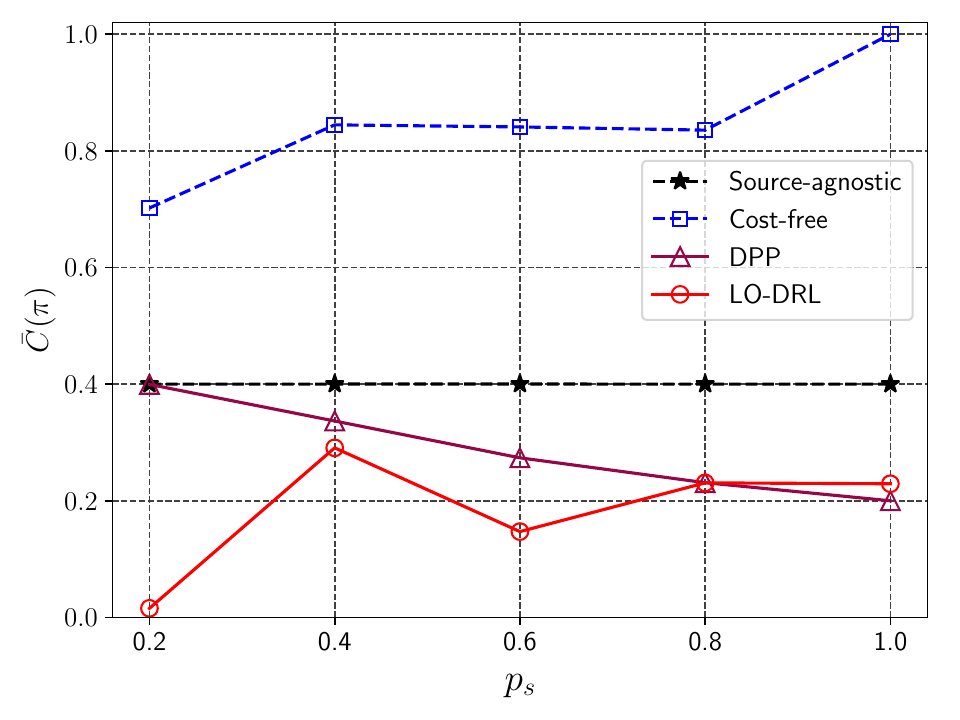}
		\caption{The transmission frequency vs. $p_s$ for the different policies.}
		\label{fig:transdiffps}
	\end{minipage}
\end{figure*}

\subsection{DRL-based Policy}
DRL is a promising tool to solve MDPs with unknown system statistics and large state/action spaces. However, time-averaged constraint satisfaction is a non-trivial issue in the DRL framework. One approach is to design an iterative training procedure that finds the optimal Lagrangian multiplier $\lambda^*$ and its corresponding $\lambda$-optimal policy $\pi_{\lambda^*}$. However, this approach may suffer from high computational complexity. 

As it can be seen in \eqref{eq:dpp-1}, the one-slot expected drift-plus-penalty function jointly considers cost minimization and constraint satisfaction. This means it can be used to regulate the behavior of a DRL agent and guide it toward a constraint-satisfying solution. Inspired by this, we propose a Lyapunov optimization-based DRL (LO-DRL) policy in the following.

\subsubsection{MDP formulation} At each slot $t$, the DRL agent observes current system state $s_t$, execute an action $\alpha_t \in \mathcal{A}$ according to the stationary policy $\pi(\cdot|s_t)$. Then the system state transitions to next state $s_{t+1}$ with a probability $P(s_{t+1}|s_t, \alpha_t)$. The agent learns policy through rewards, i.e., $R(s_t, \alpha_t, s_{t+1})$. The goal of the DRL agent is to maximize the average expected reward over an infinite horizon, i.e.,
\begin{align}
	J(\pi) = \limsup\limits_{T\rightarrow \infty}\frac{1}{T}\sum_{t=0}^{T-1} \mathbb{E}^\pi\{\gamma^t R_t\}
\end{align}
where $\gamma$ is the discounted factor that trades off long-term against short-term performance. In this work, we consider the negative drift-plus-penalty function as the reward signal, i.e.,
\begin{align}
	R_t = - (L(\mathcal{Z}_{t+1}) - L(\mathcal{Z}_t) + V\Delta_{t}). \label{eq:reward}
\end{align}
Herein, the one-slot expectations are ignored because: 1) the DRL agent has no knowledge about source/channel statistics, and 2) the system randomnesses can be averaged out through accumulated rewards. 

\begin{remark}
    The LO-DRL is a model-free approach that does not rely on prior information about the source/channel statistics. Although the offline training time scales exponentially with the number of sources, the LO-DRL offers real-time decision-making capability after deployment.
\end{remark}



\subsubsection{Algorithm design} In this work, the DRL agent is trained using the Proximal Policy Optimization (PPO) method\cite{PPO}, which is recognized as one of the state-of-the-art algorithms for solving large-scale MDPs with discrete or continuous action spaces. We adopt an actor-critic network architecture that consists of two neural networks: an ``actor" network that optimizes the policy and a ``critic" network that evaluates the performance. Both networks have four fully-connected layers with ReLU activation functions. The input layer has $2M$ neurons, where $M$ is the number of sources. The size of the hidden layers is $(128, 128)$. The output layer has $M+1$ neurons for the actor network, representing the action probabilities, and one neuron for the critic network, representing the state value.

\section{Simulation Results}\label{sec:results}
In this section, we validate the performance of the proposed policies. In our setup, we consider unit resource consumption cost, i.e., $c_m = 1, m \in \mathcal{M}$. Therefore, $\bar{C}(\pi) \in (0, 1]$ represents the transmission frequency. We consider the following baseline policies for comparison purposes
\begin{itemize}
	\item \textbf{Source-agnostic Policy:} Probabilistic actions taken according to \eqref{eq:source-agnostic} with $c_m = 1$ and $\epsilon=0$.
	\item \textbf{Cost-free Policy:} It is a special $\lambda$-optimal policy for $\lambda = 0$ in \eqref{Problem:LagrangianMDP} which minimizes the average CAE while neglecting the transmission cost. 
\end{itemize}

For the LO-DRL policy, the number of steps per episode is $10000$, the learning rate of the actor/critic network is $0.0003/0.001$, and the discount factor $\gamma$ is 0.99. Note that we need to train different agents for different scenarios, such as different numbers of sources or different system statistics. 

In the remaining, we first validate the performance of the proposed policies in a single-source scenario and analyze their behavior under different channel conditions. The multi-source scenario is examined in Section~\ref{sec:multiple sources}.

\subsection{Performance Comparison}
We first consider a single source scenario for performance comparison of different policies. We consider a source ($S_1$) consisting of four states, and its self-transition matrix and CAE matrix are time-invariant and are shown below

\begin{align}
	P^1 = \begin{bmatrix}
		0.8 & 0.2 & 0 & 0 \\
		0.1 & 0.8 & 0.1& 0\\
		0& 0.1& 0.8& 0.1\\
		0& 0& 0.2& 0.8
	\end{bmatrix}, 
	\delta^1 = \begin{bmatrix}
		0 & 10 & 50 & 30 \\
		10& 0& 40& 20\\
		20& 10& 0& 10\\
		30& 20& 40& 0
	\end{bmatrix} \notag
\end{align}
where $\delta^1_{i,j} = \delta^1(X^1_t = i, \hat{X}^1_t = j)$. We set $C_\textrm{max} = 0.4$ and $V = 100$.

Fig.~\ref{fig:costdiffps} and Fig.~\ref{fig:transdiffps} compare the performance of different policies as a function of the success probability. We observe that the proposed DPP and LO-DRL policies far outperform the source-agnostic policy regarding average CAE minimization. The LO-DRL policy outperforms the DPP policy, especially when the channel quality is poor. This is because the LO-DRL policy takes into account the long-term system performance and is capable of learning the system dynamics without prior knowledge. Remarkably, it shows that for relatively ``good" channel conditions, we can achieve the performance of the cost-free policy by utilizing fewer transmissions but exploiting the timing of the important information. \textit{This demonstrates the effectiveness of timing and the importance of information in such systems.}

\begin{figure*}[ht]
	\centering
	\begin{minipage}[t]{0.49\textwidth}
		\centering
        \includegraphics[width=0.85\linewidth]{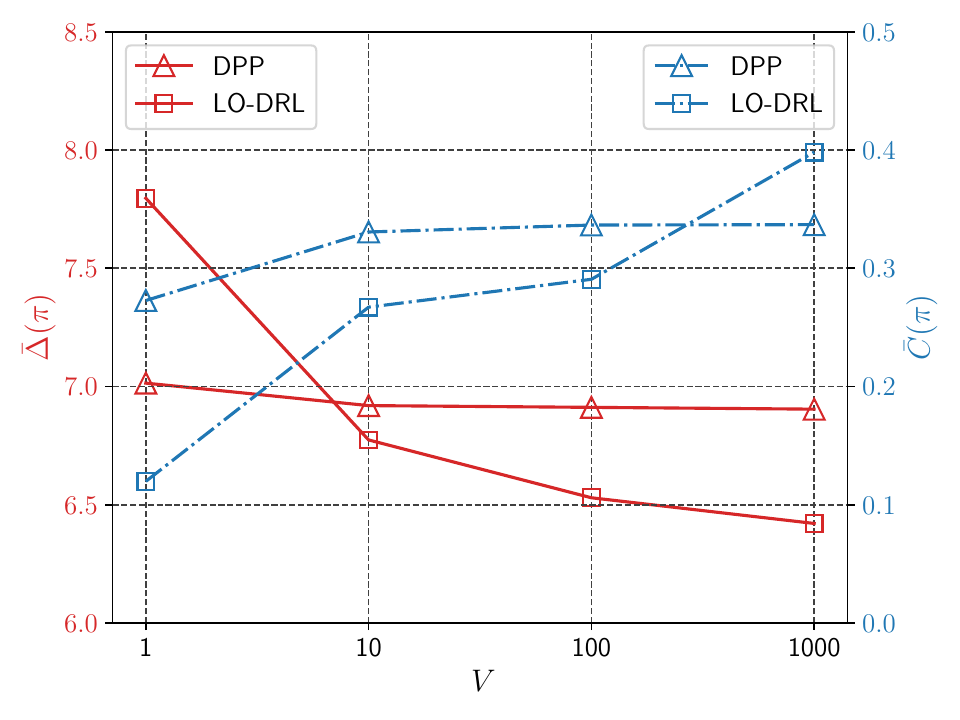}
        \caption{Performance comparison of the DPP and the LO-DRL policies.}
        \label{fig:comparision_dpp_drl}
	\end{minipage}
	\hfill
	\begin{minipage}[t]{0.49\textwidth}
		\centering
    	\includegraphics[width=0.85\linewidth]{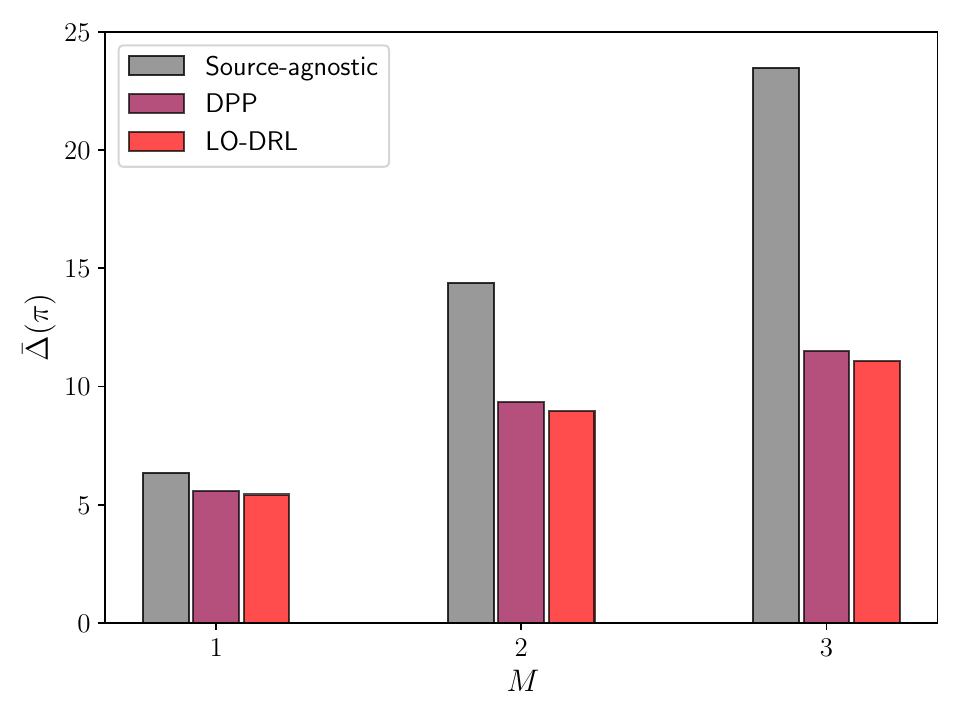}
    	\caption{Average CAE vs. different number of sources for different policies.}
    	\label{fig:multisource}
	\end{minipage}
\end{figure*}

\subsection{Sensitivity Analysis}
Fig.~\ref{fig:comparision_dpp_drl} compares the performance of the DPP and LO-DRL policies for different values of $V$. The success probability is $p_s = 0.4$, and the maximum allowed transmission cost is $C_\textrm{max} = 0.4$. It can be seen that both the DPP and LO-DRL policies satisfy the constraint, and the transmission frequency grows as $V$ increases. Additionally, the LO-DRL policy outperforms the DPP policy when $V$ is large. However, when $V$ is small, the LO-DRL policy may emphasize transmission reduction (virtual queue stability), thus resulting in performance degradation. The LO-DRL policy is sensitive to $V$ and can outperform the DPP policy with an appropriately chosen $V$.

\subsection{Multiple Sources}\label{sec:multiple sources}
We consider another type of source consisting of two states, and the self-transition matrix is $[1-p,p;q,1-q]$. We set $S_2$ as a slow-varying source with $p = 0.1, q=0.15$ and $S_3$ as a fast-varying source with $p = 0.2, q=0.7$. Additionally, the CAE matrices are $\delta^2 = \delta^3 = [0, 5; 1, 0]$. We set $c_m = 1, \omega_m=1, V = 100, p_s = 0.6, C_\textrm{max} = 0.8$.

Fig.~\ref{fig:multisource} compares the average CAE for different policies. It can be seen that our policies significantly outperform the source-agnostic policy in the considered scenarios. Due to the limitation on system resources, the average CAE grows as the number of sources increases. However, the average CAE of the DPP and LO-DRL policies grows much slower than the source-agnostic policy. This also shows that \textit{our goal-oriented policies factor in the significance of source states and the effectiveness of timing}. Furthermore, the LO-DRL policy performs better than the DPP policy in all scenarios.

\section{Conclusion}\label{sec:conclusion}
In this work, we studied the problem of remote estimation of multiple Markov sources in resource-constrained systems. We showcased how the CAE metric enables semantics-empowered and goal-oriented communication for NCSs. Furthermore, we developed two sampling policies that achieve near-optimal performance in CAE minimization while significantly reducing the ineffective status updates.

\bibliographystyle{IEEEtran}
\bibliography{ref}

\end{document}